\documentclass[11pt]{paper}
\usepackage{amsmath,amsthm,amssymb,amsfonts}
\usepackage[margin=1in]{geometry} 
\usepackage{mathpazo}
\usepackage[numbers]{natbib}
\usepackage[usenames,dvipsnames]{color}
\usepackage{hyperref}
\usepackage{verbatim}
\usepackage{amssymb}
\usepackage{float,ctable}
\usepackage{mathrsfs}
\usepackage{algorithm, algorithmic}
\usepackage{lineno}
\usepackage{cleveref}
\usepackage{xspace}
\usepackage{nicefrac}

\usepackage[final,inline,marginclue]{fixme}
\fxsetup{mode=multiuser,theme=colorsig,inlineface=\itshape,envface=\itshape}
\FXRegisterAuthor{sv}{asv}{Suresh}
\FXRegisterAuthor{aa}{aaa}{Amir} 

\theoremstyle{plain}
\newtheorem{theorem}{Theorem}[section]
\newtheorem*{theorem*}{Theorem}
\newtheorem{lemma}{Lemma}[section]
\newtheorem{corollary}{Corollary}[section]
\newtheorem{claim}{Claim}[section]

\numberwithin{equation}{section}
\newtheorem{defn}{Definition}[section]

\newcommand{\reals}{{\mathbb R}}

\newcommand{\breg}{\ensuremath{D_\phi}}

\newcommand{\eps}{\varepsilon}
\newcommand{\etal}{\emph{et al}\xspace}
\newcommand{\norm}[3]{\ensuremath{\|#1\|_{#3, #2}}}
\newcommand{\hamming}{\ensuremath{\{0,1\}}}

\usepackage{acronym}

\acrodef{ann}[ANN]{approximate nearest neighbor}
\acrodef{pm}[Partial Match]{partial match}
\acrodef{bb}[BB]{Bonami-Beckner}

\hypersetup{colorlinks=true,citecolor=BrickRed,linkcolor=ForestGreen}

\usepackage{tocloft}
\title{A directed isoperimetric inequality with application to Bregman near
  neighbor lower bounds\thanks{This research was partially supported by the
    National Science Foundation under grant CCF-0953066.}}
\author{Amirali Abdullah\\University of Utah  \and Suresh Venkatasubramanian\\University of Utah}
\date{}
\begin{document}
\begin{titlepage}
\maketitle
\thispagestyle{empty}
\begin{abstract}
Bregman divergences are important distance measures that are used in applications such as computer vision, text mining, and speech processing, and are a focus of interest in machine learning due to their information-theoretic properties.
There has been extensive study of algorithms for clustering and near neighbor search with respect to these divergences. In all cases, the guarantees depend not just on the data size $n$ and dimensionality $d$, but also on a structure constant $\mu \ge 1$ that depends solely on a generating convex function $\phi$ and can grow without bound independently. In general, this $\mu$ parametrizes the degree to which a given divergence is ``asymmetric". 

In this paper, we provide the first evidence that this dependence on $\mu$ might be intrinsic. We focus on the problem of \ac{ann} search for Bregman divergences. We show that under the cell probe model, any non-adaptive data structure (like locality-sensitive hashing) for $c$-approximate near-neighbor search that admits $r$ probes must use space $\Omega(dn^{1 + \mu/c r})$. In contrast for LSH under $\ell_1$ the best bound is $\Omega(dn^{1+ 1/cr})$.

Our result interpolates between several known lower bounds both for LSH-based \ac{ann} under $\ell_1$ as well as the generally harder \ac{pm} problem (in non-adaptive settings). The bounds match the former when $\mu$ is small and the latter when $\mu$ is 
$\Omega(d/\log n)$. This further strengthens the intuition that \ac{pm} corresponds to an ``asymmetric'' version of \ac{ann}, as well as opening up the possibility of a new line of attack for lower bounds on \ac{pm}.

Our new tool is a directed variant of the standard boolean noise operator. We prove a generalization of the Bonami-Beckner hypercontractivity inequality (restricted to certain subsets of the Hamming cube), and use this to prove the  desired \emph{directed} isoperimetric inequality that we use in our data structure lower bound.
\end{abstract}


\end{titlepage}


\section{Introduction}

Bregman divergences generalize the squared Euclidean distance and standard
projective duality. They include well studied distance measures like the
Kullback-Leibler divergence, the Itakura-Saito distance, bit entropy and the exponential distance, and
appear naturally as distance  functions for data analysis.  

Bregman divergences retain many of the \emph{combinatorial}
properties of $\ell_2$ and so \emph{exact} geometric algorithms based on space decomposition
(Voronoi diagrams, convex hulls and so on) can be used to compute the
corresponding Bregman counterparts~\cite{bvd}. But the
divergences are asymmetric\footnote{In fact the squared Euclidean distance is the \emph{unique} symmetric Bregman divergence.} and violate triangle
inequality, and so break most approximation algorithms for distance problems
(clustering, near neighbor search and the like) that make heavy use of these properties.

This ``degree of violation'' can be quantified as a scalar parameter $\mu$ that
depends only on the functional form of the divergence (not the size of input or its dimension). There are many ways $\mu$ is defined in the literature~\cite{abdullah2012,bregworst,ackermann2010}, and these are all loosely related to the view of $\mu$ as a measure of \emph{asymmetry}: given a Bregman divergence $D$ over a domain $\Delta$, define $\mu$ as $\max_{x,y \in \Delta} \frac{D(x,y)}{D(y,x)}$. 

To the best of our knowledge, $\mu$
appears as a term in theoretical guarantees for \emph{all} constant factor approximation algorithms for
geometric problems in these spaces. This is highly unsatisfactory because $\mu$ can grow without bound independent of the data size or dimensionality. It is therefore natural to ask the question: 




\begin{quote}
\emph{  Is this dependence on $\mu$ intrinsic ? Or are there clever algorithms that
  can circumvent the effect of asymmetry for such problems ?
}\end{quote}

In this paper we provide the first evidence that this dependence is indeed
intrinsic under a broad range of the parameters $n$ and $d$ (namely $d \gg \log n$).
 We focus on the fundamental problem of \ac{ann} search, which has been studied extensively for Bregman divergences. 

We show the following under the cell probe model for \emph{uniform} Bregman divergences (loosely speaking, distances composed as a sum of $d$ identical measures):
\begin{theorem}\label{thm:main}
For a uniform Bregman divergence $D$ with measure of asymmetry $\mu$ in each dimension, let $L = \min \left( \frac{d}{\log n}, \mu \right)$. Any non-adaptive data structure which in $r$ probes can return even a $c'$ approximation to the nearest neighbor under $D$ with 
constant probability (over the choice of query) requires $\Omega(dn^{1 + \Omega(L/c'r)})$ space.
\end{theorem}

In particular, this lower bound applies to methods based on locality-sensitive hashing and to several popularly used divergences such as the Kullback-Leibler or Itakura-Saito distances. Note that in comparison to the space lower bound of 
$ \Omega(dn^{1 + \Omega(1/cr)})$ for Euclidean (or $\ell_1$) \ac{ann}~\cite{geometric}  which is sub-quadratic and near linear for sufficiently high $c$,  the space lower bound here is polynomial in $n$ with an exponent of $\Omega(\mu)$ for constant factor approximations and (as we show later) strengthens upto $\mu = \Theta(d / \log n)$.\footnote{Note that since there are $2^d$ points on the cube, we must have that $d > \log n$ just to fit all the point set in the cube.} This indicates the increased hardness introduced by asymmetry. 

\paragraph{\ac{ann} and \ac{pm}.}

There is one aspect of our work that may be of independent interest. Separately from our main result, we can show a direct reduction from geometric problems on the Hamming cube to the equivalent problems for Bregman divergences.In Section 
\ref{sec:partial} we find a very interesting ``interpolation'' of lower bounds parametrized by $\mu$: a constant factor approximation for Bregman \ac{ann} with $\mu = O(1)$ implies a constant factor approximation for $\ac{ann}$ under $\ell_1$, and a similar approximation for Bregman $\ac{ann}$ with $\mu = \Omega(d)$ implies a constant factor approximation for \ac{pm}, which is notoriously hard problem. Intriguingly while lower bounds for \ac{pm} are in general higher than those of \ac{ann}, at the intermediate point $\mu = \Theta (\frac{d}{\log n})$ in the interpolation we already obtain lower bounds that are as strong as those known for \ac{pm} (with the qualifier that our analysis restricts to non-adaptive algorithms.

One interpretation of this is that $\mu$ captures the intuition that \ac{pm} is an ``asymmetric'' version of \ac{ann}. It would also be interesting if this directed perspective allows us to obtain improved lower bounds for \ac{pm} itself by a reduction in the opposite direction. Indeed in the strictly linear space regime, the lower bound of $\Omega(d)$ queries for our asymmetric \ac{ann} is stronger than those known for \ac{pm} ($\Omega \left(d / \log d \right)$ for adaptive algorithms by \cite{directsum}
).



\subsection{Overview of our approach.}
\label{sec:overview-our-methods}

Our approach makes use of the Fourier-analytic approach to proving lower bounds
for (randomized) near-neighbor data structures that has been utilized in a
number of prior works~\cite{geometric,motwanilower,metricexpand}.  This
approach generally works as follows: one thinks of the purported data structure
as a partition of the Hamming cube, and in particular as a function defined on
the Hamming cube. Then one shows that any such function is ``expansive'' with
respect to small perturbations: in effect, that points scatter all over the
cube. As a consequence, probing any particular cell of a data structure does not
yield enough useful information because of the scattering, and one has to make
many probes to be sure.  The key technical result is showing that the function
is expansive, and this is done using Fourier-analytic machinery, and
hypercontractivity of the noise operator in particular~\cite{ryan2013}.  One also needs to
construct a "gap instance" where the gap between nearest neighbor and second
nearest neighbor is large.

While black-box reductions from $\ell_1$-\ac{ann} can yield weak lower bounds for Bregman divergences (see Section~\ref{sec:reduct-from-hamm}), we need a much stronger argument to get a $\mu$-sensitive bound. Specifically, we need the following components:
\begin{itemize}
\item \textbf{A gap instance}: We create an instance that separates a near neighbor
  at distance $\eps d$ from a second nearest neighbor at distance $\mu d$. To do
  so, we define a Bregman hypercube and associated asymmetric noise operator
  (with different probabilities of changing $0$ to $1$ and $1$ to $0$) and observe our gap is far stronger than the natural symmetric analog - $\Omega \left(\frac{\mu}{\eps} \right)$ 
vs $\Omega \left(\frac{1}{\eps} \right)$.
\item \textbf{Directed  hypercontractivity}: The Fourier-analytic machinery breaks down for our noise
  operator because of lack of symmetry. Indeed, a simple example shows that a
  natural directed analog of the \ac{bb} inequality cannot be true. Instead, we prove a directed \ac{bb} inequality in Section~\ref{sec:asymm-bregm-cube} that is true "on average", or
  on a subset of the hypercube, which will be sufficient for our lower bound. We
  prove this by relating the norm of the directed noise operator to related norms on biased (but symmetric) measure spaces, 
allowing us to make use of \ac{bb}-type inequalities in these spaces.  
\item \textbf{A scatter lemma}: Showing that points
  "scatter" is relatively easy in symmetric spaces: in the directed setting, the argument can be made in a 
similar way but requires a nontrivial analysis of associated collision rates and inner products which we carry out in Section~\ref{sec:shattering-property}.
\item \textbf{An information-theoretic argument}: We borrow the argument used by
  \cite{geometric}. Essentially, the scatter lemma shows a small sampling of the cells of a successful data structure must resolve many query points and thus will have high information content. This allows us to lower bound the space required by such a structure in Section~\ref{sec:from-hyperc-lower} and obtain Theorem \ref{thm:main}.
\end{itemize}


\section{Related Work}
\label{sec:related-work}
Bregman distances were first introduced by Bregman~\cite{bregman}. They are the unique divergences that satisfy certain axiom systems for distance measures~\cite{csiszar}, and are key players in the theory of information geometry~\cite{amari}. Bregman distances are used extensively in machine learning, where they have been used to unify boosting with different loss functions~\cite{collins2002logistic} and unify different mixture-model density estimation problems~\cite{dhillon}. A first study of the algorithmic geometry of Bregman divergences was performed by Boissonnat, Nielsen, and Nock~\cite{bvd}. They observed since Bregman divergences retain the same combinatorial structures as $\ell_2$, many exact algorithms from the Euclidean domain carry over naturally with the same bounds. For example, they showed that  exact near neighbors can be computed in $O(n^{\lfloor \frac{d}{2} \rfloor})$ via a Voronoi diagram. Nielsen and Nock also observed that the smallest enclosing disk can be computed exactly in polynomial time~\cite{nock2005fitting}.
 
As discussed earlier, these parallels do \emph{not} carry over to the approximate setting with the lack of a triangle inequality and symmetry rendering most tools for algorithm design useless. The algorithms that do exist attempt a work around via a structure constant $\mu$. This constant is at least $1$, and grows larger as the space becomes increasingly non-metric. There are many algorithms for clustering whose resources are parametrized by $\mu$:  M\"{a}nthey and Roglin~\cite{roglin1} compute approximate $k$-means with an extra  $\mu^6$ factor under a certain perturbation model. Ackermann and  Bl\"{o}mer \cite{musimilarcoresets} exhibit a  $ O\left(\mu^2 \log k \right)$-approximate solution to $k$-means clustering via a $k$-means$++$-like procedure.  The same authors give a $O(\mu)$ approximate $k$-median clustering for a certain class of well behaved input instances~\cite{ackermann3}.  McGregor and Chaudhuri \cite{mcgregor}  avoid dependence on $\mu$ in an approximate algorithm for $k$-means clustering under the KL-divergence, but at the cost of a $\log(n)$ factor in approximation. They also show $k$-means is NP hard to approximate within a constant factor if the centers are restricted to be from the point set and implicitly leverage the non-metric nature of the space in their bound. 

For \ac{ann} search, Abdullah, Moeller and Venkatasubramanian~\cite{abdullah2012} gave an algorithm that is efficient in constant dimensions.  Their algorithm yields a $1+\eps$ approximate nearest neighbor with an additional dependence on $\mu^{O(d)}$ besides standard dependence on factors of $\frac{1}{\eps^{O(d)}}$ and $\log n$. Indeed, this paper is a consequence of attempting to extend their results to higher dimensions. 

There are numerous heuristic algorithms for computing with Bregman divergences approximately, including algorithms for the minimum enclosing ball~\cite{nsmallestdisk} and near neighbor search~\cite{cayton-thesis,bregsearch}.

Lower bounds for near neighbor search in metric spaces have been studied extensively. Borodin, Ostrovsky and Rabani~\cite{emnn} show a lower bound that any randomized cell probe algorithm for the exact match problem that must probe at least $\Omega(\log d)$ cells. Barkol and Rabani improve this bound to $\Omega(\frac{d}{\log n})$ cells~\cite{2000tight}. Liu~\cite{liu} proves a lower bound of $d^{1 - o(1)}$ on the query time of a deterministic approximate nearest neighbor algorithm in the cell probe model, whereas Chakrabarti and Regev give a lower bound of $\Omega \left( \frac{ \log \log d}{\log \log \log d} \right)$  for the randomized case~\cite{chakraOptimal}.

Our work is in the spirit of the program initiated by Motwani, Naor and Panigrahy~\cite{motwanilower}, who analyze a random walk in the Hamming cube to lower bound the LSH quality parameter $\rho$ as $\frac{1}{2c}$ ($c$ is the separation between near and far points). O'Donnell, Wu and Zhou~\cite{O'Donnell:2014:OLB:2600088.2578221} later tighten this to $\frac{1}{c}$. Panigrahy, Talwar and Wieder~\cite{geometric} use the Boolean noise operator to simulate perturbations on the Hamming cube, and use hypercontractivity to show that these Hamming balls touch many cells of a data structure and obtain space-query trade off cell probe lower bounds. They then extend these to broader classes of metric spaces with certain isoperimetric properties of vertex and edge expansion~\cite{metricexpand}.  This is not a comprehensive survey;~\cite{metricexpand} give a good overview of several of the known lower bounds. 

All of the above approaches use Fourier analysis on boolean functions over the hypercube. This is a vast literature that we will not survey here: the reader is pointed to Ryan O'Donnell's lecture notes~\cite{ryan2013}. In particular, we make use of results by Keller~\cite{kellers} and Ahlberg \etal~\cite{biasedcontract} on the analysis of the noise operator in biased spaces.

\section{Bregman Divergences And The Bregman Cube}
\label{sec:bregman-divergences}

We start with some definitions.  Let $\phi: M\subset \reals^d \to \reals$ be a \emph{strictly convex} function that is differentiable in the relative interior of $M$. The \emph{Bregman divergence} $\breg$ is defined as 
\[ \breg(x,y) = \phi(x) - \phi(y) - \langle \nabla \phi(y), x-y\rangle. \]
An important subclass of Bregman divergences are the \emph{decomposable} Bregman divergences. Suppose $\phi$ has domain $M = \prod_{i=1}^d M_i $ and can be written as $\phi(x) = \sum_{i=1}^d \phi_i(x_i)$, where $\phi_i :M_i \subset \reals \to \reals$ is also strictly convex and differentiable in relint($M_i$). Then 
\[ \breg(x,y) = \sum_{i=1}^d D_{\phi_i}(x_i, y_i) \]
is a \emph{decomposable} Bregman divergence. 

 Table \ref{tab:breg-examples} illustrates some of the commonly used ones.
\ctable[
    caption = Commonly used Bregman divergences ,
    label = tab:breg-examples,
    pos = htbp
]{c|c|c|c}{
\tnote{The Mahalanobis distance is technically not decomposable, but is a linear transformation of a decomposable distance}
\tnote[b]{($S^d_{++}$ denotes the cone of positive definite matrices)}
}{
    Name & Domain & $\phi$ & \breg(x,y)\\ \hline
    $\ell_2^2$ & $\reals^d$ & $\frac{1}{2}\|x\|^2$  & $\frac{1}{2}\|x - y \|^2_2$ \\
    Mahalanobis\tmark & $\reals^d$ & $\frac{1}{2} x^\top Q x$& $\frac{1}{2}(x-y)^\top Q (x-y)$\\
    Kullback-Leibler & $\reals^d_+$& $\sum_i x_i \log x_i$& $\sum x_i \log \frac{x_i}{y_i} - x_i + y_i$\\
    Itakura-Saito & $\reals^d_+$& $-\sum_i \log x_i$& $\sum \Bigl( \frac{x_i}{y_i} - \log \frac{x_i}{y_i} - 1\Bigr)$\\
    Exponential & $\reals^d$& $\sum_i e^{x_i}$& $\sum e^{x_i} - (x_i - y_i +1)e^{y_i} $\\
    Bit entropy & $[0,1]^d$ & $\sum_i x_i \log x_i + (1-x_i) \log(1-x_i)$ & $\sum x_i \log \frac{x_i}{y_i} + (1-x_i) \log \frac{1-x_i}{1-y_i} $\\
    Log-det & $S^d_{++}\tmark[b]$& $\log \det X$& $\langle X, Y^{-1} \rangle - \log \det XY^{-1} - N $\\
    von Neumann entropy & $S^d_{++}$& $\text{tr} (X \log X - X)$& $ \text{tr}(X (\log X - \log Y) - X + Y)$\\ \hline
}

Finally, we define the special case of a \emph{uniform} Bregman divergence which is a decomposable $\breg$ where all the $\phi_i, 1 \leq i \leq d$ are identical. In this case, we simply refer to each $\phi_i$ as $\phi_{\reals} : \reals \to \reals$
and we have that $\phi(x) = \sum_{i=1}^d \phi_{\reals}(x_i)$. Note that most commonly used Bregman divergences are uniform, including the Kullback-Leibler, Itakura-Saito, Exponential distance and Bit entropy.  In what follows we will limit ourselves to uniform Bregman divergences.

\paragraph{Quantifying asymmetry.}
\label{sec:quant-asymm}

It is clear from the definition of $\breg$ that in general $\breg(x,y) \ne
\breg(y,x)$. In what follows, we define the \emph{measure of asymmetry} as $\mu = \max_{x,y \in M} \frac{\breg(x,y)}{\breg(y,x)}$.

By construction, $\mu \ge 1$. But it is not arbitrary: rather, it is a function
of the generating convex function $\phi$ and the domain over which it is defined. 
To see this, note that the Bregman divergence $\breg(x,y)$ can be viewed as the error
(evaluated at $x$) incurred in replacing $\phi$ by its first-order approximation
$\tilde{\phi}(x) = \phi(y) + \langle \nabla \phi(y), x - y\rangle$. By the
Lagrange mean-value theorem, this error can be written as the quadratic form
$\breg(x,y) = \langle x-y, \nabla^2 \phi(c) (x-y)\rangle$ where $\nabla^2\phi$ is the Hessian
associated with $\phi$, and $c = c(x,y)$ is some point on the line connecting $x$ and
$y$. Note that this point $c$ will general be different to the point $c'$ that
achieves equality when measuring $\breg(y,x)$. 

Thus $\frac{\breg(x,y)}{\breg(y,x)}$ is bounded by the ratio of the maximum to minimum eigenvalue that the Hessian $\nabla^2\phi$ realizes over the domain $M$. In particular, $\mu$ need not depend on the number of points
$n$ or the dimension $d$.  

Most prior work on algorithms with Bregman divergences focus on violations of the
triangle inequality, rather than symmetry. However, the different variants of
$\mu$ defined there all relate in similar fashion to the ratio of eigenvalues of
the Hessian of $\phi$, and can be shown to be loosely equivalent to each other; in the sense that if the measure of asymmetry $\mu$ grows without bound, so do these measures.

\paragraph{The Bregman Cube.}
\label{sec:bregman-cube}

We introduce a new structure, the \emph{Bregman cube} $B_{\phi} =
\{0, 1 \}^d$ alongwith asymmetric distance measure $D$.  This is combinatorially equivalent to a regular Hamming cube, but
where distances of $1$ and $\mu$ are associated with flipping a bit from $1$ to
$0$ and $0$ to $1$ respectively. More precisely, given $D: \{0,1\}^d \times \{0,1\}^d \to \reals$ and asymmetry parameter $\mu$, we stipulate:

\begin{equation}
D(x,y) = \mu |\{i : y_i > x_i \} | +  | \{j : x_j > y_j \} |,\forall x,y \in \{0,1 \}^d .
\end{equation}
We note now how $B_{\phi}$ and the associated measure $D$ can be induced from a uniform Bregman divergence $D_{\phi}$ on 
$\reals^d$. Let the asymmetry parameter $\mu$ of $D_{\phi_{\reals}}$ be realized by points $a$, $b \in \reals$. W.l.o.g. (due to scaling) assume $D_{\phi_{\reals}} (b, a) =1$ and $D_{\phi_{\reals}}(a,b) = \mu$. Then distances on $B_{\phi}$  with parameter $\mu$ correspond exactly to those on $\{a , b \}^d \subset \reals^d$ under $D_{\phi}$.  
\begin{figure}[H]
  \begin{center}
    \includegraphics[scale = 0.4]{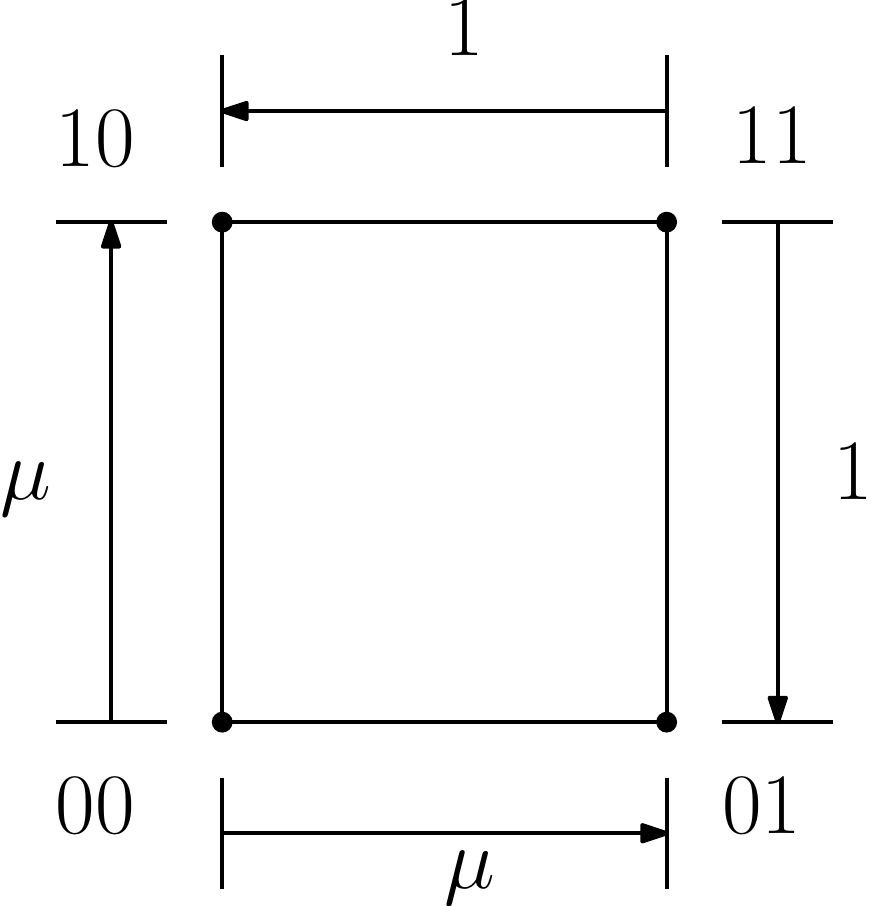}
  \end{center}
  \caption{Asymmetric distances.}
  \label{cube}
\end{figure} 
We use standard notation to define a ``$c$-approximate nearest neighbor" ($c$-ANN) for query point $q$ and point set $P$ under $D$. Namely, let $p' \in P$ be a ``$c$-approximate nearest neighbor" to $q$ if $D(q,p') \leq c \min_{p \in P} D(q,p)$. We also fix $q$ to be the first argument in the asymmetric distance $D$ to maintain consistency.

\section{Preliminaries of Fourier analysis}
\label{sec:prel-four-analys}
\subsection{Basis and Fourier coefficients.}

Let the \emph{$p$-biased measure} $\kappa_p = (p \delta_{\{1\}} + (1-p)\delta_{\{0\}})^{\otimes d}$ be the product measure defined over the hypercube $\{0,1\}^d$. Note that for $p = \nicefrac{1}{2}$ this is the uniform measure over binary strings of length $d$.  All expectations and norms are implicitly defined according to the choice of measure $\kappa_p$ as follows:

For any function $f: \{0,1\}^d \to \reals$, let 
\begin{enumerate}
\item{$E_p[f] = \sum_{x \in \{0,1\}^d} \kappa_p(x) f(x)$.}
\item{$\norm{f}{p}{j} = \left( \sum_{x \in \{0,1\}^d} \kappa_p(x) f(x)^j  \right)^{1/j}$.}
\end{enumerate}
It is well known that there is a natural Fourier basis for the space $F_p$ of all functions $f: \{0,1\}^d \to \reals$ with respect to $\kappa_p$ (see for example, \cite{biasedcontract, kellers, ryan2013}). For each $x \in \{0,1\}^d$ and $i \in [d]$ let
 \[\chi_i^p =  \begin{cases} 
      \sqrt{\frac{p}{1-p}} & x_i = 0 \\
      -\sqrt{\frac{1-p}{p}} & x_i = 1 \\
   \end{cases}
\]
The set of $\chi_i^p$ corresponds to a bit wise parity basis which we can extend to arbitrary $S \subset [n]$ as 
$\chi_S^p(x) = \prod_{i \in S} \chi_i^p(x)$. The resulting $\chi_S^p$ form an orthonormal basis of  $F_p$. That is, 
we can define the Fourier coefficient corresponding to a $S \subset [n]$ as:

\begin{equation}
\hat{f}^{(p)}(S) = \sum_{x \in \{0,1\}^d} \kappa_p(x) f(x) \chi_S^p(x).
\end{equation}

And hence obtain that
\begin{equation}\label{paritybasis}
f = \sum_{S \subseteq [n]} \hat{f}^{(p)} (S) \chi_S^p.
\end{equation}

The orthonormality of the $\chi_S^p$ immediately yields the Parseval identity 
\begin{equation}\label{parsev}
\norm{f}{p}{2}^2 = E_p[f^2] = \sum_{S \subseteq [n]} \left( \hat{f}^{(p)}(S) \right)^2. 
\end{equation}

See \cite{ryan2013} for a full discussion of the parity basis. We note that wherever we drop the superscript and simply write $\hat{f}(S)$ or $\chi_S$, we intend $\hat{f}^{(1/2)}(S)$ and $\chi_S^{\frac{1}{2}}$ respectively.  

\subsection{Noise operator and hypercontractivity.}
For $x \in \{0,1\}^d$, let $y$ be the random variable obtained by flipping each
bit of $x$ with probability $p$. The \emph{noise operator} $T_{\delta} f$ for a
function $f$ is defined as the expectation of $f$ over $y$ of $T_{\delta} f(x) = E_y[f(y)]$:
\[ T_{\delta} f(x) = E_y[f(y)]. \]
In the case of the uniform measure $\kappa_{\frac{1}{2}}$, $T_{\delta} f$ can be
written as $T_{\delta} f = \sum_{S} (1-2 \delta)^{|S|} \hat{f}^{(1/2)}(S) \chi_S$.

More generally, given a function $f$ and choice of measure $\kappa_p$ we define the operator
\begin{equation}
\label{eq:2}
\tau_{\delta} f = \sum_{S} \delta^{|S|} \hat{f}^p (S) \chi_S.
\end{equation}
And we note that $T_{\delta} = \tau_{1-2 \delta}$ for the uniform measure.
\begin{theorem}[Hypercontractivity\cite{ryan2013}]
\label{standardhyper}
\begin{equation}
\norm{\tau_\delta f}{\frac{1}{2}}{2} \leq \norm{f}{\frac{1}{2}}{1+\delta^2}.
\end{equation}
\end{theorem}
The above result holds for the uniform measure space ($p = \nicefrac{1}{2}$) but can also be considered in 
general $p$-biased measure spaces. The hypercontractivity problem in this context is to find $C(p,\delta)$ such that 
$\norm{\tau_\delta f}{p}{2} \leq \norm{f}{p}{1+ C(p,\delta) }$. Partial results were obtained by Talagrand~\cite{talagrand}, Friedgut
~\cite{friedgut} and Kindler~\cite{kindler}, whereas stronger bounds were obtained by Diaconis and Saloff-Coste~\cite{diaconis} and the optimal known value of $C(p,\delta)$ was obtained by Oleskiewicz~\cite{olesz}. 

We prefer the following formulation of a bound on $C(p,\delta)$ by Keller~\cite{kellers} due to convenience in some algebraic cancellations:
\begin{theorem}\label{biasedhyper}
Let $\bar{p} = min(p, 1-p)$, where $0 \leq p \leq 1$. Then for any $ \delta \geq 0$, s.t.  $\delta^2 \sqrt{\frac{\bar{p} \lfloor \log \frac{1}{\bar{p}} \rfloor}{1-\bar{p}}} \leq 1$, we have \footnote{In the remainder of the paper we drop the floor arguments, which do not affect any of the asymptotics of our result.}:
\begin{equation}
\norm{\tau_\delta f}{p}{2} \leq \norm{f}{p}{1+ \delta^2 \left(1-\bar{p} \right) / \left(\bar{p} \lfloor \log 1/\bar{p} \rfloor \right) }.
\end{equation}
\end{theorem}
 Observe that if we set $p = \nicefrac{1}{2}$ the general expression described in Theorem
\ref{biasedhyper} reduces to the special case of Theorem \ref{standardhyper}. We also note that the constants in our main result of Theorem~\ref{thm:main} may improve slightly if we use, for instance, the optimal value of the hypercontractivity parameter from Oleskiewicz~\cite{olesz}; however the asymptotics are unaffected. 
\section{Isoperimetry in the directed hypercube}
\label{sec:asymm-bregm-cube}

\subsection{The asymmetric noise operator.}

For any point $x \in \hamming^d$, let $\nu_{p_1, p_2}(x)$ be the distribution
obtained by independently flipping each $0$ bit of $x$ to $1$ with
probability $p_1$ and each $1$ bit of $x$ to $0$ with probability $p_2$. 

\begin{defn}[Asymmetric noise operator]
  The asymmetric noise operator $R_{p_1, p_2}$ is an operator defined on
  functions over $\{0,1\}^d$ and is defined as
\[[R_{p_1, p_2} f](x) = E_{y \sim \nu_{p_1, p_2}(x) }[f(y)]. \]
\end{defn}
We note that Benjamini, Kalai and Schramm~\cite{benjamin} study a version of this asymmetric noise in the context of percolation crossings, and that the formulation of $R_{p,0}$ by \citet{biasedcontract} as a Fourier operator is highly useful in our analysis.
We observe first that if we set $p = p_1 = p_2$, then $R_{p_1, p_2} = T_p$. There is in fact a
stronger relationship between the two operators. 
\begin{theorem}\label{estructure}
If $p_1 \geq p_2$, $p_1 \leq 1 -  p_2$ then $R_{p_1,p_2} = T_{p_2} R_{\frac{p_1-p_2}{1-2 p_2},0}. $
\end{theorem}
\begin{proof}
The overall transition probabilities from a $0$ to a $1$ and vice versa must
match on both sides of the equation. Therefore if we set $R_{p_1,p_2} = T_{p'} R_{p'',0}$, then the following
two equations must hold true. 
\begin{align*}
p_2 &= p' \\
p_1 &= (p'')(1-p') + (1-p'')p'
\end{align*}
Solving this system yields us $p' = p_2$ and $p'' = \frac{p_1-p_2}{1-2 p_2}$. 
\end{proof}

Our goal is to prove hypercontractivity for $R_{p_1, p_2}$. By the decomposition
given in Theorem~\ref{estructure} and known hypercontractivity bounds for $T_p$,
it will suffice to study how $R_{p,0}$ affects the Fourier coefficients of $f$.  This turns out to be intimately related
to the $p$-biased measure $\kappa_p$. We will combine this with standard hypercontractivity
results for $T_p$ to obtain the desired bound. (Bounds for $R_{0,p}$ follow by easy analogy and are also presented, although not needed for our main results.) Since we are looking at asymptotic bounds, it will be best to think of both $p_1$ and 
$p_2$ as smaller than a fixed constant, say $\frac{1}{100}$. 

\subsection{Hypercontractivity of $R_{p,0}$.}
\label{sec:hyperc-r_p-0}
Suppose we are given a function $f : \{ 0,1 \}^d \to \reals$. 

\citet[Lemma 4.2]{biasedcontract} show that the Fourier coefficients of the \emph{asymmetric}
perturbation of $f$ in a \emph{uniform} space are related to the
Fourier coefficients of $f$ in a \emph{biased} space. (For the sake of completeness, we present their proof in 
Section \ref{sec:ahlberg}.)

\begin{theorem}[\citep{biasedcontract}]\label{thm:transf}
  \begin{align*}
    \widehat{R_{0,p}f}^{\left(1/2 \right)}(S) &= \left( \sqrt{\frac{1-p}{1+p}} \right)^{|S|} \hat{f}^{ \left(\frac{1-p}{2} \right)}(S).\\
    \widehat{R_{p,0} f}^{\left( 1/2 \right)}(S) & = \left(
      \sqrt{\frac{1-p}{1+p}} \right)^{|S|} \hat{f}^{ \left(\frac{1+p}{2}
      \right)}(S).
  \end{align*}
\label{cited}
\end{theorem}

Using this, we obtain the following result relating the asymmetric operator
$R_{p,0}$ and $R_{0,p}$ to the symmetric operator $\tau_\delta$ in a biased space.

\begin{theorem}\label{ineq1}
\begin{equation}
\norm{R_{0,p} f}{\frac{1}{2}}{2} = \norm{\tau_{\sqrt{\frac{1-p}{1+p}}}
  f}{\frac{1-p}{2}}{2}.
\end{equation}
\begin{equation}\label{eq:1}
\norm{R_{p,0} f}{\frac{1}{2}}{2} = \norm{\tau_{\sqrt{\frac{1-p}{1+p}}} f}{\frac{1+p}{2}}{2}.
\end{equation}
\end{theorem}

\begin{proof}
  The proof follows by combining Parseval's identity with the definition of $\tau_\delta$ in Equation~(\ref{eq:2}).
\end{proof}

Theorem \ref{ineq1} does \emph{not} directly imply hypercontractivity for $R_{p,0}$
under the uniform measure. Instead it relates the $l_2$ norm of $R_{p,0}f$ to
the norm of $f$ in a biased measure space. Indeed there can be adversarial
choices of $f$ where the norm increases under perturbation by $R_{p,0}$. 

\paragraph{Example.}
Consider the function $f:\{0,1\} \to \reals$. Let $f(0) = 0$ and $f(1) = 1$.
Then $R_{p,0} f (0) = p$ and $R_{p,0} f(1) = 1$.   In particular $\|R_{p,0}
f \|_{2,1/2}^2 = \frac{p^2+1}{2} > \|f \|_{2, 1/2} = \frac{1}{2}$ which indicates no hypercontractivity.

We address this issue in two parts.  Firstly, we use the biased Bonami-Beckner
inequality (Theorem \ref{biasedhyper}) to relate the right-hand side of Eq.~(\ref{eq:1}) to
the norm of $f$. 

\begin{theorem}\label{ineq2}
\begin{equation} \norm{\tau_{\sqrt{\frac{1-p}{1+p}}} f}{ \frac{1-p}{2}}{ 2 } \leq \norm{f}{ \frac{1-p}{2}}{ 1 + \frac{1}{1 - \log (1-p)}   }. \end{equation}
\begin{equation} \norm{\tau_{\sqrt{\frac{1-p}{1+p}}} f}{ \frac{1+p}{2}}{ 2 } \leq \norm{f }{ \frac{1+p}{2}}{ 1 + \frac{1}{1 - \log (1-p)}   }. \end{equation}
\end{theorem}

\begin{proof}
We recall first the statement of Theorem \ref{ineq1}: $\norm{R_{p,0} f}{ \frac{1}{2}}{ 2 } = \norm{\tau_{\sqrt{\frac{1-p}{1+p}}} f}{ \frac{1+p}{2}}{ 2 }.$ We combine this with the biased hypercontractivity claim of Theorem~\ref{biasedhyper} which states:
\begin{equation*}
\norm{\tau_\delta f}{ \bar{p}}{ 2 } \leq \norm{f}{ \bar{p}}{ 1+ \frac{1-\bar{p}}{\bar{p} \lfloor \log \frac{1}{\bar{p}} \rfloor}\delta^2   }.
\end{equation*}
We note the $\bar{p}$ there is the smaller of the measures of $0$ or $1$ in the product space. Hence we plug in $\bar{p} = \frac{1-p}{2}$  and $\delta = \sqrt{\frac{1-p}{1+p}}$ to obtain:

\begin{align*}
& 1+  \delta^2 \frac{1-\bar{p}}{\bar{p}  \log \frac{1}{\bar{p}} }   \\
= & 1 + \left( \sqrt \frac{1-p}{1+p} \right)^2 \left( 1 - \frac{1 - p}{2} \right)   / \left(\frac{1-p}{2}  \log \left(1 / \frac{1-p}{2} \right) \right) \\
= & 1 + \left(\frac{1-p}{1+p} \right) \left( \frac{1+p}{1-p} \right) /  \frac{1}{\log \frac{2}{1-p}} \\
= & 1 + \frac{1}{1 - \log (1-p)}.
\end{align*}
The second result claimed in the theorem statement follows almost identically.
\end{proof}
The second and final part of the argument is to relate the norms of $f$ in the
unbiased and biased spaces. Recall that our ultimate aim is to bound $\norm{R_{p,0}f}{ \frac{1}{2}}{2} $ by $\norm{ f}{ \frac{1}{2}}{ 1 + \frac{1}{1- \log(1-p)} }$. 

Let us 
limit $f : \{0, 1 \}^d \to \{0 , 1 \}$ to take its support from the lower half of the Hamming cube, which we stipulate as $L = \{ x : 
\sum_i x_i \leq \frac{1}{2} d \}$. We can define the upper half of the Hamming cube $L$ analogously. Whenever we refer to a function $f_U$, this will be understood to have support only on the upper half of the Hamming cube, whereas $f_L$ will have support only on the lower half. 

\begin{theorem}\label{ineq3}
For any parameters $\delta > 1$  and $\frac{1}{2} \leq p \leq 1$, we have:
\begin{equation}
     \norm{f_L}{ p}{ \delta }^{\delta}  \leq  \norm{f_L}{ \frac{1}{2}}{ \delta }^{\delta}.
\end{equation}
And for $0 \leq p \leq \frac{1}{2}$, we have:

\begin{equation}
     \norm{f_U}{  p}{ \delta }^{\delta}  \leq  \norm{f_U}{ \frac{1}{2}}{ \delta }^{\delta}.
\end{equation}
\end{theorem}
\begin{proof}
The first inequality follows because points in the lower half of the hypercube have larger measure under the uniform distribution than under the $p$-biased distribution for $p > \frac{1}{2}$. The second claimed inequality follows by symmetry.
\end{proof}

By Theorems \ref{ineq1}, \ref{ineq2} and \ref{ineq3} we finally obtain that:
\begin{theorem}
\begin{equation}
  \norm{R_{0,p} f_U}{ \frac{1}{2} }{ 2 }   \leq    \norm{f_U}{ \frac{1}{2}}{ 1+ \frac{1}{1- \log (1-p)} } .     
\end{equation}

\begin{equation}
  \norm{R_{p , 0} f_L}{ \frac{1}{2} }{ 2 }  \leq    \norm{f_L}{ \frac{1}{2}}{ 1+ \frac{1}{1- \log (1-p)} }.     
\end{equation}
\end{theorem}

We will find the following asymptotic form of our result useful, and indeed this is our main tool employed in Section \ref{sec:shattering-property}.
\begin{corollary}\label{cor:asymp}
\begin{equation}
  \norm{R_{0,p} f_U}{ \frac{1}{2} }{ 2 }   \leq    \norm{f_U}{ \frac{1}{2}}{ 2 - p \log e + O(p^2)} .     
\end{equation}

\begin{equation}
  \norm{R_{p , 0} f_L}{ \frac{1}{2} }{ 2 }  \leq    \norm{f_L}{ \frac{1}{2}}{ 2 - p \log e + O(p^2)}.     
\end{equation}
\end{corollary}
\begin{proof}
By employing the Taylor expansions of $\log(1-x)$ and $1/(1-x)$.
\end{proof}

We can now generalize to the case of $R_{p_1, p_2}$. Let $p_1 \geq p_2$ , and define $f_L$ as before. We do not use the following theorem in the remainder of this paper, but we include it for the interested reader who seeks a complete statement of the 
hypercontractivity result. 

\begin{theorem}\label{bracketed}
For $p_1 \geq p_2$, and both $p_1, p_2 \leq \frac{1}{4}$, we have that:
\begin{equation}\label{mineq5}
  \norm{R_{p_1,p_2} f_L}{\frac{1}{2} }{ 2 }   \leq    \norm{f_L }{ \frac{1}{2}}{ 1+ (1-2p_2)^2 / \left(1- \log \left(\frac{1-p_1-p_2}{1-2 p_2} \right) \right) } .     
\end{equation}
\end{theorem}
\begin{proof}
First note by Theorem \ref{estructure}, we have $
R_{p_1,p_2} = T_{p_2} R_{\frac{p_1-p_2}{1-2 p_2},0}.$
Now let $p = \frac{p_1 - p_2}{1- 2 p_2}$. Recalling Theorem \ref{thm:transf} and that $T_{p_2} = \tau_{1-2p_2}$, we obtain:

\begin{equation}\label{mineq1}
\norm{ T_{p_2} R_{p,0} f_L}{ \frac{1}{2}}{2 } = \norm{\tau_{(1-2 p_2)}\tau_{\sqrt{\frac{1-p}{1+p}}} f_L}{ \frac{1+p}{2}}{ 2 }.
\end{equation}
Now using the fact that $\tau_{a} \tau_{b} = \tau_{ab}$ , and by a similar calculation as in Theorem \ref{ineq2} for hypercontractivity in a biased measure space, we obtain :

\begin{equation}\label{mineq2} 
\norm{\ \tau_{(1-2p_2) \sqrt{\frac{1-p}{1+p}}} f_L}{ \frac{1+p}{2}}{ 2 } \leq \norm{f_L }{ \frac{1+p}{2}}{1 + \frac{(1-2p_2)^2}{1 - \log (1-p)}}. \end{equation}
We can combine equation \ref{mineq2} with Theorem \ref{ineq3} to get:
\begin{equation}\label{mineq4}
 \norm{T_{p_2} R_{p , 0} f_L}{ \frac{1}{2} }{ 2 }   \leq    \norm{f_L}{ \frac{1}{2}}{ 1+ \frac{(1-2p_2)^2}{1- \log (1-p)} } .     
\end{equation}
Now substituting back the value of $p = \frac{p_1-p_2}{1-2 p_2}$ into equation \ref{mineq4} we obtain the claimed result.
\end{proof}

\section{Hard Input Distributions for the Bregman Cube}
\label{sec:analysis-r_p_1-p_2}

We now describe the construction of a hard input distribution for the Bregman
cube. The key properties of this distribution will be that a query point will
(in expectation) either have a near neighbor within distance $O(\eps d)$, or will not have any
neighbor closer than $\Omega(\mu d)$. Note that in contrast, the corresponding gap
distribution for the Hamming cube via the \emph{symmetric} noise operator has a gap of $O(\eps d)$ for the
nearest neighbor versus $\Omega(d)$ for the second nearest neighbor. Finally for the purposes of our result and this paper we will assume $\mu \geq \frac{1}{\eps}$. 

\subsection{Generating our input and query on the cube.}\label{inputquery}
Define a random perturbation $\nu : \{0,1 \}^d \to \{0,1\}^d$ as
a random binary string $\nu_{p_1, p_2}(x)$ obtained by flipping any $0$ bit in $x$ to $1$ with probability $p_1$ and a $1$ bit to $0$
with probability $p_2$. 
In what follows, assume that $p_1 = \eps < \frac{1}{100}$ and $p_2  = \frac{\eps}{\mu}$. 

Uniformly at random pick $n$ elements $S= \{s_1,s_2, \ldots s_n \}$
from the set $L = \{x \in \{0,1 \}^d $ s.t $\sum(x) \leq \frac{d}{2} \}$, which
is the lower half of the Bregman cube.
We first perturb $S$ to obtain $P = \nu_{\eps/\mu, \eps}(S)$.  We then perturb $S$ in the opposite direction, to 
obtain $Q = \nu_{\eps , \eps/\mu} (S)$. We now assign $P$ to be our data set and choose our query point $q$ uniformly at random from $Q$.

\begin{theorem}\label{thm:generating-our-input}
Let $q$ be the perturbation of $s \in S$, and $p \in P$ be the corresponding
point of $P$. Then for $\mu = O \left(\frac{\eps d}{\log n} \right)$  and  with high probability (at least $1 - 1/\text{poly}(n)$):
\begin{enumerate}
\item{
 $\forall p' \neq p$, $p' \in P$, $D(q, p')= \Omega(\mu d)$.}
\item{$D(q, p) = \Theta(\eps d)$.}
\item{$\forall p' \neq p$, $p' \in P$ , $ \frac{D(q,p')}{D(q,p)}  = \Omega(\mu/\eps)$.}
\end{enumerate}
\end{theorem}

\begin{proof}
 We focus on the distance induced by a single bit, and multiply by $d$ to get the overall distance(follows from each bit being chosen identically and independently). 

For the first claim, we show that $D(q_i, p'_i) = \Omega(\mu)$ with high probability. To aid our argument, we define the Hamming weight of a point as $H(x) = \sum_i(x_i)$. Now Chernoff concentration bounds give us that if $C = \{x \in U | 0.5 
\leq H(x) \leq 0.55 \}$, then $|C|/|U|  \geq 1 - e^{- \Omega(d)}$. Therefore each bit of randomly chosen $p' \in U$ is $0$ with probability at least $0.5 - 1 / \text{poly n}$. We obtain similarly that $q_i$ is $1$ with probability at least 
$0.5 - 1 / \text{poly}(n)$ for $\eps$ smaller than a suitable choice of constant.

 Since $D(0,1) = \mu$  and $q$ is independent of $p'$ for $p' \neq p$, we can argue now that $E[D(q,p')] = \Omega(\mu d)$. A standard Chernoff analysis shows that $D(q,p') = \Omega(\mu d)$ holds true for all $p' \neq p$ with high probability.

We consider now the second claim. Refer to the $j$-th bit of $s$ , $q$ and $p$ as $s^j$, $q^j$ and $p^j$ respectively and recall again that $D(0,1) = \mu$ and $D(1,0) = 1$. Consider first the case where $s^j = 0$. Then:
\begin{align*}
E[D(q^j, p^j) | s^j = 0] 
=& \Pr[q^j = 1 | s^j = 0 ] \Pr[p^j = 0 | s^j = 0] D(1,0)\   +  \\
& \Pr[q^j = 0 | s^j = 0] \Pr[p^j = 1 | s^j = 0] D(0,1)   \\
=& \eps \left(1 - \frac{\eps}{\mu} \right) D(1,0) +  (1 - \eps) \left( \frac{\eps}{\mu} \right) D(0,1) \\
=& \eps \left(1 - \frac{\eps}{\mu} \right) + (1 - \eps) \left( \frac{\eps}{\mu} \right)( \mu) \\ 
=&  2 \eps - \eps^2 - \frac{\eps^2}{\mu} = \Theta(\eps).
 \end{align*}

Similarly when $s^j = 1$, we have:
\begin{align*}
E[D(q^j, p^j) | s^j = 1] 
=& \Pr[q^j = 1 | s^j = 1 ] \Pr[p^j = 0 | s^j = 1] D(1,0)  + \\
& \Pr[q^j = 0 | s^j = 1] \Pr[p^j = 1 | s^j = 1] D(0,1)   \\
=& \  \left(1- \frac{\eps}{\mu} \right)  \left(\eps \right)  D(1,0) + \frac{\eps}{\mu} \left(1- \eps \right) D(0,1) \\
=& \  \left(1- \frac{\eps}{\mu} \right)  \left(\eps \right)   + \frac{\eps}{\mu} \left(1- \eps \right) \mu \\
=&\  2 \eps - \eps^2 - \frac{\eps^2}{\mu} = \Theta(\eps).
\end{align*}

We show now these distances concentrate around the expectation. Recall the classic Chernoff bounds : given a collection of independent $0$-$1$ random variables $X_i$ and $X = \sum_i X_i$ such that $u_x$ is the mean of $X$, then $Pr[|X - u_x|  \geq \sigma u_x ] \leq e^{\frac{-\sigma^2 u_x}{2}} + e^{\frac{-\sigma^2 u_x}{3}}$. Note that here that we can represent our distances 
in the form $Y = \sum_{i=1}^d Y_i$ and $Z = \sum_{i=1}^d Z_i$ where $i$ is an index over the number of bits $d$ and the probability of success is $\frac{\eps}{\mu} \left(1 - \eps  \right)$ and $\eps \left(1 - \frac{\eps}{\mu} \right)$ respectively. (Or asymptotically $\eps$ and $\frac{\eps}{\mu}$ respectively.)  For $Y$ and $Z$ to concentrate around $u_y$ and $u_z$ respectively for all $n$ points and suitable choice of constant $\sigma$, we clearly require $u_y = \Omega( \log n )$ and $u_z = \Omega(\log n)$. This requires $\frac{\eps}{\mu} d = \Omega(\log n)$. 

And finally, the third claim can be seen to follow directly from the first two.
\end{proof}
We note that it can be shown even for arbitrarily large $\mu$ and some constant $\eps$, that a ``gap'' of $\Omega(\frac{d}{ \log n})$ can be achieved by setting $\mu ' 
= d / \log n$ and applying perturbations $P = \nu_{\eps/\mu', \eps}(S)$ and $Q = \nu_{\eps , \eps/\mu'} (S)$ respectively.

\section{Shattering A Query}
\label{sec:shattering-property}
We are now ready to assemble the parts that make up the proof of
\Cref{thm:main}. In this section, we show that a point ``shatters'': namely,
that if we perturb a point by a little, then it is likely to go to many
different hash buckets. In the next section, we will show that this implies an
information-theoretic lower bound on the number of queries needed to recover the
original point that generates the query. 

We prove the shattering bound in two steps. In Lemma~\ref{lem:shattering-query}
we show that if we fix any sufficiently small subset of the cube, then the set
of points that are likely to fall into this subset under perturbation is
small. Then in Lemma~\ref{modularshatter}, we use this lemma to conclude that
for \emph{any} partition of the space into sufficiently small sets (think of
each set as the entries mapped to a specific hash table entry), any perturbed
query will be sent to many of these sets (or equivalently, no entry contains
more than a small fraction of the ``ball'' around the query). 

As mentioned in Section~\ref{sec:overview-our-methods}, the structure of this
section mirrors the argument presented by Panigrahy
\etal\cite{geometric}. The difficulty is that we can no longer directly work
with the (symmetric) operator $T_\eps$, and so the analysis becomes more
intricate. 
	
We consider sets of points restricted to $L = \{x : \sum_i x_i \leq \frac{d}{2} \}$.
Let $A$ be a light cell, where we stipulate a light cell to be such that $|A| \leq a 2^d$ for some small $0 < a <
1$. Define $\gamma_{y, \eps, \frac{\eps}{\mu}} (A) = \Pr[\nu_{\eps,
  \frac{\eps}{\mu}}(y) \in A]$.  Let $B \subseteq L$ be the set of points for
which a perturbation is likely to fall in $A$, i.e. $B = \{y \in L \mid  \gamma_{y, \eps, \frac{\eps}{\mu}} (A)  \geq a^{c_0 \eps} \}$ for some $0 < c_0 <1$ to be chosen later.  

We shall show that $|B| \leq 2^d a^{1+ c_1 \eps}$, where once again $0 < c_1 < 1$ can be set later.  To this purpose, we will use 
Taylor approximations to simplify the algebra to asymptotic behavior.

\begin{lemma}
\label{lem:shattering-query}
Let $A \subseteq \{0,1 \}^d$ with $|A| \leq a.2^d$. Let $\eps \in (0 , 1)$, $\mu \geq \frac{1}{\eps}$ and
$B = \{y \in L \mid  \gamma_{y, \eps, \frac{\eps}{\mu}} (A) \geq a^{c_0 \eps} \}$. Then for suitable choices of constants $c_0$ and $c_1$ less than $1$, and for sufficiently small $\eps$.
\begin{equation}
|B| < 2^d a^{1 + c_1 \eps}.
\end{equation} 
\end{lemma}
\begin{proof}
Suppose on the contrary that $|B| > 2^d a^{1 + c_1 \eps}$. By definition, for every $y \in B$,  $Pr[\nu_{\eps, \frac{\eps}{\mu}} (y) \in A] \geq a^{c_0 \eps}$. Let $Q_B$ denote the random variable obtained by picking an 
element from $B$ uniformly at random and then applying $\nu_{\eps, \frac{\eps}{\mu}}$.

Now, 
\begin{align*}
Pr[Q_B \in A] &= \frac{2^d}{|B|} \langle R_{\eps, \frac{\eps}{\mu}} 1_B  , 1_A \rangle   \text{  (By definition of $R_{\eps, \frac{\eps}{\mu}}$.) }\\
&= \frac{2^d}{|B|}  \langle T_{\frac{\eps}{\mu}} R_{ \left(\eps - \frac{\eps}{\mu} \right) / \left(1 - 2 \frac{\eps}{\mu}  \right) , 0} 1_B , 1_A \rangle  \text{  ( By Theorem \ref{estructure} ). } \\
& = \frac{2^d}{|B|} \langle \tau_{1 - 2 \frac{\eps}{\mu}} R_{\left(\eps - \frac{\eps}{\mu} \right) / \left(1 - 2 \frac{\eps}{\mu}  \right) , 0} 1_B , 1_A \rangle \text{  (By the definition of $\tau$ ) }  
\end{align*}

For convenience, let $\delta_1 =  \sqrt{ \left(1 - \frac{\eps - \frac{\eps}{\mu}}{1 - 2 \frac{\eps}{\mu}}\right) \Big/ \left(1 + \frac{\eps - \frac{\eps}{\mu}}{1 - 2 \frac{\eps}{\mu}} \right)}$, $\delta_2 = 1 - 2 \frac{\eps}{\mu}$ and $p = \left(\eps - \frac{\eps}{\mu} \right) / \left(1 - 2 \frac{\eps}{\mu}  \right)$. 

We abuse notation slightly, and introduce the function:
\begin{equation}
1_B^{\left(\frac{1+p}{2}\right)}  = \sum_{S \subset \{0, 1 \}^d} \hat{1_B}^{\frac{1+p}{2}}(S) \chi_S^{\frac{1}{2}}.
\end{equation}

 I.e., this is a function whose Fourier coefficients in the uniform measure are the Fourier coefficients of $1_B$ in the $\frac{1+p}{2}$ measure. Now by Theorem 
\ref{thm:transf}, we can transform  $R_{\left(\eps - \frac{\eps}{\mu} \right) / \left(1 - 2 \frac{\eps}{\mu}  \right) , 0} 1_B$ as $\tau_{\delta_1} 
1_B^{(\frac{1+p}{2})}$:

\begin{align*}
Pr[Q_B \in A] &= \frac{2^d}{|B|} \langle  \tau_{\delta_2} \tau_{\delta_1} 1_B^{(\frac{1+p}{2})} , 1_A \rangle \\
&= \frac{2^d}{|B|} \langle \tau_{\delta_1 \delta_2} 1_B^{(\frac{1+p}{2})} , 1_A \rangle  \text{ (Since $\tau$ is multiplicative.) } \\
&= \frac{2^d}{|B|} \langle \tau_{(\delta_1 \delta_2)^{3/4}} 1_B^{(\frac{1+p}{2})} , \tau_{(\delta_1 \delta_2)^{1/4}} 1_A \rangle \text{ (Since $\tau$ can be distributed in a dot product.) } 
\end{align*}

We now proceed using Cauchy Schwarz to upper bound the dot product as a product of two norms, Parseval's to claim $\| \tau_{(\delta_1 \delta_2)^{3/4}} 1_B^{p'}\|_{2,\frac{1}{2}} = \| \tau_{(\delta_1 \delta_2)^{3/4}} 1_B \|_{2,p'}$ and then hypercontractivity to upper bound each of the norms in biased and uniform measure space respectively. Setting $p' = \frac{1+p}{2}$:
\begin{align*} 
 Pr[Q_B \in A] 
  & \leq \frac{2^d}{|B|}\| \tau_{(\delta_1 \delta_2)^{3/4}} 1_B^{p'}\|_{2,1/2} \|\tau_{(\delta_1 \delta_2)^{1/4}} 1_A \|_{2 , 1/2}   \\
&= \frac{2^d}{|B|}\| \tau_{(\delta_1 \delta_2)^{3/4}} 1_B \|_{2,p'} \|\tau_{(\delta_1 \delta_2)^{1/4}} 1_A \|_{2 , 1/2}
     \\
 & \leq \frac{2^d}{|B|}\|  1_B \|_{1+ \frac{p'}{(1-p') \log (1/(1-p')) }(\delta_1 \delta_2)^{1.5}, p'}  \| 1_A \|_{1+ 
\sqrt{\delta_1 \delta_2}, 1/2} .
\end{align*}

And finally, since $B \subset L$, by \ref{ineq3} we have $ \|  1_B \|_{1+ \frac{p'}{(1-p') \log (1/(1-p')) }(\delta_1 \delta_2)^{1.5}, p'}$ will only increase if we measure the norm in uniform space instead of the biased space with $p' > \frac{1}{2}$:  
\begin{equation}\label{lastinter}
Pr[Q_B \in A] \leq \frac{2^d}{|B|} \|  1_B \|_{1+ \frac{p'}{(1-p') \log (1/(1-p')) }(\delta_1 \delta_2)^{1.5}, 1/2}  \| 1_A \|_{1+ 
\sqrt{\delta_1 \delta_2}, 1/2}
\end{equation} 

On a high level now, our approach is simply to show that the power in the norm on both expressions is $2 - \Omega(\eps)$. This would show that the collision or intersection size $Pr[Q_B \in A] \leq \frac{2^d}{|B|}$ is much smaller than the product of the set sizes. To simplify these expressions, we focus now on $\frac{p'}{(1-p') \log (1/(1-p')) }(\delta_1 \delta_2)^{1.5}$ and $\sqrt{\delta_1 \delta_2}$. For $\sqrt{\delta_1 \delta_2}$, some straightforward substitution and the assumption that $\mu \geq \frac{1}{\eps}$ shows:
\begin{align*}
\sqrt{\delta_1 \delta_2} &= \sqrt{\left(1 - 2 \frac{\eps}{\mu} \right) \sqrt{\frac{1 - \eps - \frac{\eps}{\mu}}{1+ \eps - 3 \frac{\eps}{\mu}}}} \\
&= 1 - \frac{\eps}{2}+ O(\eps^2).
\end{align*}

We come now to $\frac{p'}{(1-p') \log (1/(1-p') )}(\delta_1 \delta_2)^{1.5}$. Simply plugging in the values for $p'$, $\delta_1$ and $\delta_2$ now shows:

\begin{align*}
\frac{p'}{(1-p') \log (1/(1-p')) } (\delta_1\delta_2)^{1.5} &=  \left(  \frac{1+\eps - 3 \frac{\eps}{\mu}} { 1 - \eps - \frac{\eps}{\mu}}    \right)^{1/4}   \frac{\left( 1 - 2 \frac{\eps}{\mu} \right)^{1.5}}{\left(1 - \log \left( 1 - \left (
\eps - \frac{\eps}{\mu}  \right) / \left( 1 - 2 \frac{\eps}{\mu} \right) \right) \right) }  \\
 &<  \left(  1 + 2 \eps + O(\eps^2)   \right)^{1/4}   \frac{1}{1 - \log \left(1 - \eps + O(\eps^2) \right)  }   \\
& \leq \left(1 + \frac{\eps}{2} + O(\eps^2) \right)  \left(1 - \eps \log e + O( \eps^2)  \right)\\
& < \left(1 - \frac{\eps}{2} + O(\eps^2) \right). 
\end{align*}
Or for sufficiently small $\eps$, that $\frac{p'}{(1-p') \log (1/(1-p')) } (\delta_1\delta_2)^{1.5} = 1 -  \Omega(\eps)$. 
We can now see the asymptotic behavior of the norms in Equation \ref{lastinter}: 
\begin{equation}\label{eq:splitA}
Pr[Q_B \in A] \leq \frac{2^d}{|B|} \|  1_B \|_{2 - \Omega(\eps), 1/2}  \| 1_A \|_{2 - \Omega(\eps), 1/2}
\end{equation}
Hence there exists a constant $k$,  such that:
\begin{align*}
Pr[Q_B \in A] & \leq \frac{2^d}{|B|} \| 1_B \|_{2 - k \eps, 1/2}  \| 1_A \|_{2 - k \eps, 1/2} .
\end{align*}

Now let $k \eps =2  \eps'$, and also set $c_0 \eps$ and $c_1 \eps$ to be $\frac{\eps'}{6}$. By the assumptions of our lemma, we have $\Pr[Q_B \in A] \geq a^{\frac{\eps'}{6}}$.  Recalling that $\| 1_B \|_{2 - 2 \eps', \frac{1}{2}} = |B|^{1 /(2-2 \eps')}$ and $\| 1_A \|_{2 - 2 \eps' , \frac{1}{2}} = |A|^{1/(2-2 \eps')}$ , we obtain :

\begin{align*}
Pr[Q_B \in A] &\leq \frac{2^d}{|B|} \|  1_B \|_{2 - 2 \eps', 1/2}  \| 1_A \|_{2 - 2 \eps', 1/2}  \\
&= \frac{2^d}{|B|} \left( \frac{|B|}{2^d} \right)^{1/(2- 2 \eps')}
 \left( \frac{|A|}{2^d} \right)^{1/(2- 2 \eps')}  \\
&= \left ( \frac{|A|}{2^d} \right)^{1/(2-2\eps')} \left( \frac{|B|}{2^d} \right)^{1/(2-2\eps') -1}. 
\end{align*}
 Now recalling that $|A| \leq a 2^d$ and claiming for contradiction that 
$|B| \geq 2^d a^{1+ \eps'/6} $, we obtain:

\begin{align*}
Pr[Q_B \in A] &\leq  a^{1/(2-2 \eps')} a^{(1 + \eps'/6)(\frac{1}{2 - 2 \eps'} - 1)} \leq a^{\eps'/6}.
\end{align*}
However this is impossible, since by our lemma assumptions $Pr[Q_B \in A] \geq  a^{\eps'/6}$.
 Hence we must have that if $\Pr[Q_B \in A] \geq a^{\eps'/6}$ then $|B| \leq 2^d a^{1 + \eps'/6}$. Noting that $\eps' = k \eps$ and $k$ is some constant,  our lemma follows.
\end{proof}
The following lemma is now an easy consequence.
\begin{lemma}\label{modularshatter}
Let $A_1, \ldots , A_m$ be partitions of $\{0,1 \}^d$ and let $LC = \{i : \|A_i \| \leq 2^d / \sqrt{m} \}$ be the set of light cells. Then:
\begin{equation}
\Pr_{y \in LC} [ \max_{i \in LC} \gamma_{y, \eps, \frac{\eps}{\mu}} (A_i) \geq m^{- \frac{c_0 \eps}{2}}] < m^{- \frac{c_1 \eps}{2}}.
\end{equation}
\end{lemma}

\begin{proof}
Let $a_i = \frac{|A|}{2^d}$ and note $\sum_i a_i = 1$.  By Lemma \ref{lem:shattering-query} $\Pr_{y \in LC} [  \gamma_{y, \eps, \frac{\eps}{\mu}} (A_i) \geq a_i^{ c_0 \eps}] \leq a_i^{1+c_1 \eps} $]. And we also have by the bound on light cells that $a_i^{ c_0 \eps} \leq m^{-\frac{c_0 \eps}{2}}$.  Then by a union bound, we have that the desired probability is:
\begin{align*}
&\Pr_{y \in LC} [ \max_{i \in LC} \gamma_{y, \eps, \frac{\eps}{\mu}} (A_i) \geq m^{- \frac{c_0 \eps}{2}}] 
 \leq  \sum_{i} \Pr_{y \in LC} [  \gamma_{y, \eps, \frac{\eps}{\mu}} (A_i) \geq a_i^{ c_0 \eps}] \\
 & \leq  a_i^{1+c_1 \eps} \leq \max_{i} a_i^{c_1 \eps} \sum_i a_i  \leq \max_i a_i^{c_1 \eps}  \leq (\sqrt{m})^{-c_1 \eps}.
\end{align*}
\end{proof}

\section{From hypercontractivity to a lower bound}
\label{sec:from-hyperc-lower}

First we lay out the notation and preliminaries of our argument, including our model. 
An $(m,r,w)$ non-adaptive algorithm is an algorithm in which given $n$ input points $p_1, \ldots, p_n$ in $\{0,1\}^d$ we prepare in preprocessing a table 
$T$ which consists of $m$ words, each $w$ bits long. Given a query point $q$, the algorithm queries the table at most $r$ times and let $l_1, l_2, \ldots ,l_r $ denote the set of indices looked up by the algorithm. For every $t \leq r$, the location of the $t$-th probe, $l_t = l_t(q)$ depends only upon the query point $q$ and no t upon the content that was read in the previous queries. In other words, the functions $l_1, l_2, \ldots , l_r$ depend on $q$ only. In this section, we show a time-space cell probe lower bound. We mostly use the machinery given in \cite{geometric}, but for the sake of explication and clarity we reproduce the argument and expand some of the steps.

The high level idea of \cite{geometric}, and also work by Larsen \cite{larsencellsampling} and  Wang and Yin \cite{wangyin} is ``cell sampling" of a data structure $T$ on input $P$. If $T$ resolves a large number of queries which do not err in few probes, then a small sample of the cells will resolve many queries with high probability. Now if such a sample of cells can be described in fewer bits than the information complexity of these queries, then there would be a contradiction. This lower bounds the size of $T$.

We prepare our dataset and query point as described in Section \ref{sec:analysis-r_p_1-p_2}. 
 We first pick a set of $n$ elements $S= \{s_1,s_2, \ldots s_n \}$
from the lower half of the Bregman cube. More precisely, we 
pick from the set of strings 
$U = \{x \in \{0,1 \}^d $ s.t $\sum(x) \leq \frac{d}{2} \}$. We first perturb $S$ to obtain $P = \nu_{\frac{\eps}{\mu}, \eps}(S)$.   
We then pick $i$ uniformly at random from $[1 \ldots n]$ and set $q$ to be $\nu_{\eps, \frac{\eps}{\mu}}(s_i)$. 
In what follows, we will let $s_i$ denote the point of $S$ which is perturbed to obtain $q$ and $p_i$ be the corresponding point in $P$. Theorem~\ref{thm:generating-our-input} guarantees that $D(q,p_i) = \Theta(\eps d)$ whereas 
$D(q,p_j) = \Omega(\mu d)$ for $j \neq i$ with high probability. 
Hence recovering a $\frac{\mu}{\eps}$ nearest neighbor to $q$ from $P$ is equivalent to recovering $p_i$ exactly.  The table is populated in preprocessing based on $P$ as the ground set.

Our assumption on the correctness of the algorithm is that when the input is sampled in this way, then for each $i$ with probability 
$\frac{1}{2}$ over the choice of $s_i$ and $p_i$, with probability $\frac{2}{3}$ over the choice of $q_i$ the algorithm can reconstruct $p_i$. We can fix the coin tosses of the algorithm and assume the algorithm is deterministic, and we assume the query algorithm is given access to not only $P$ but also $S$.

The main result of this section is as follows.

\begin{theorem}
\label{sec:from-hyperc-lower-1}
A $(m,r,w)$ non adaptive algorithm to recover a $O(\frac{\mu}{\eps})$ nearest neighbor to a query point $q$ with constant probability has $mw \geq \Omega(\eps d n^{1+ \Omega(\frac{\eps}{r})})$ as long as $w$ is polynomial in $d, \log n$.
\end{theorem}

Note that \Cref{sec:from-hyperc-lower-1} yields \Cref{thm:main} as a corollary, by setting $c' = \mu/\epsilon$.

\begin{proof}
We recall some standard information theoretic notation. Let $H(A)$ be the entropy of a distribution $A$, and let $I(A,B)$ be the mutual information of two distributions, such that $I(A,B) = H(A) - H(A|B) = H(B) -H(B|A)$. We have the following well-known rules for simplifying expressions:

\begin{enumerate}
\item{$I(X,Y) = I(Y,X)$.}
\item{$I(X;Y|Z) = H(X|Z) - H(X|Y,Z)$. }
\item{$I(X;Y) = 0$ if $X$ and $Y$ are independent random variables.}
\item{$I(Z;X,Y)  = I (Z;X ) + I(Z|X;Y)$. If $X$ and $Y$ are independent.}
\item{For $n$ independent random variables $X_1$ through $X_n$, we have $I(Y; X_{1:n}) = \sum_i I( Y| X_{1:i-1}; X_i)$.}

\end{enumerate}
Now let $L$ be a set of $k$ locations picked at random from the table, where $k$ is a parameter to be fixed later, and $T[L] = \{T[i] : i \in L  \}$ be the corresponding set of words.  

\begin{claim}
$I(T[L]; p_i | S, L, q) = I(T[L] ; p_i|S, L)$.
\end{claim}
\begin{proof}
When $S$ and $L$ are fixed, $p_i$ is independent of $q$, and $T[L]$ is determined by $P$. So we may simply drop $q$ here.
\end{proof}
For the remainder of our proof, for ease of notation we will implicitly assume
$S$ and $L$ are known to the algorithm and $p_i$ is conditioned on them. 
\begin{claim}
$\sum_{i=1}^n I(T[L];p_i) \leq H(T[L]) \leq wk.$
\end{claim}
\begin{proof}
First note that $H(T[L]) \leq wk$ simply from the fact that there are only $wk$ bits in $T[L]$, and $H(T[L]) \leq wk$.  For $\sum_{i=1}^n I(T[L];p_i)$, we note that $I(T[L];p_i) \leq I(T[L] | p_1,p_2, \ldots p_{i-1} ; p_i)$.
To note this, we see the direct comparison:
\begin{align*}
&I(T[L];p_i) - I(T[L] | p_1,p_2, \ldots p_{i-1} ; p_i) \\
=& H(p_i) - H(p_i|T[L]) - H(p_i ) + H(p_i| T[L],p_1,p_2, \ldots p_{i-1}) \\
=&  H(p_i| T[L],p_1,p_2, \ldots p_{i-1}) - H(p_i|T[L]) \\
\leq& 0.
\end{align*}
Hence $\sum_{i=1}^n I(T[L];p_i) \leq \sum_{i=1}^n I(T[L] | p_1,p_2, \ldots p_{i-1} ; p_i)$. Applying the chain rule, this 
latter quantity equals  $I(T[L]; p_1, p_2, \ldots p_n ) \leq H(T[L])$.
\end{proof}
Taking expectations on both sides w.r.t. $L$ and $S$, we have:

\begin{equation}\sum_{i=1}^n E_{L,S} [ I (H(T[L] ; p_i |S,L )] \leq wk. \end{equation}
Set $k = m/ n^{\Omega(\frac{\eps}{r})}$. Our goal therefore is to show that $E_{L,S}[I(T[L];p_i |S,L)] \in \Omega(\eps d)$, as this
would immediately imply the theorem.  

We will prove the slightly stronger result that $I(T[L];p_i|S, L) = \Omega(\eps d)$. Suppose that our algorithm can reconstruct $p_i$ given $T[L]$ with constant probability $\alpha$. We can lower bound $H(p_i|S, L)$ as follows. Note that it suffices to examine $H(p_i|s_i)$. In each $1$-bit of $s_i$ there is an induced entropy 
of $-\eps \log \eps - (1-\eps) \log (1- \eps) \geq -\eps \log(1-\eps) -(1-\eps) \log (1- \eps) = - \log (1- \eps) = \Omega(\eps)$. Similarly, in each $0$-bit of $s_i$, the entropy is $\frac{\eps}{\mu} \log \frac{\eps}{\mu} + (1 - \frac{\eps}{\mu}) \log ( 1 - \frac{\eps}{\mu}) = \Omega(\frac{\eps}{\mu})$. Note that $H(p_i | S)$ is at least  $\Omega(\eps d)$,  just from the entropy of the
$0$-bits of $s_i$.  

We now use the following simplification of Fano's inequality, which in slightly different form was described by Regev \cite{regev2013entropy}:

\begin{lemma}
  Let $X$ be a random variable, and let $Y = g(X)$ where $g(\cdot)$ is a random process. Assume the existence of a procedure $f$ that given $y = g(x)$ can reconstruct $x$ with probability $p$. Then 
\[ I(X;Y)  \ge p H(X) - H(p) \].
\end{lemma}
\begin{proof}
  The proof follows from the usual statement of Fano's inequality in terms of the conditional entropy $H(X|Y)$. 
\end{proof}

Consider now the mutual information $I(T[L];p_i|S, L)$. By Fano's inequality and the lower bound on $H(p_i|S, L)$, the desired lower bound on $I(T[L];p_i|S, L)$ will follow if we can present a procedure that with constant probability will reconstruct $p_i$ from $T[L]$ given $S$ and $L$. Note that $i$ is fixed in this process. 

Denote by $l_j(q)$ the location of the $j$-th query when the query point is $q$. We write $l_{[r]}$ to denote $l_1(q) \cup \ldots \cup l_r(q)$. We say a point $q_j$ is \emph{good} for $p_j$ if  $D(q_j,s_j)$ is at most $\eps d$ and $p_j$ can be reconstructed 
from $q_j$, $s_j$ and $T[L_{[r]} (q_j)]$ (the set of table lookups on $q$) with constant probability. Let $Q_i$ denote the set of points which are good for $p_i$. We make the following useful observations about $Q_i$.
\begin{lemma}\label{lem:goodProp}
\begin{enumerate}
\item{With probability at least $\frac{1}{2}$ over the choice of $s_i$, we have $\Pr [\nu_{\eps, \frac{\eps}{\mu}} (s_i)  \in Q_i] \geq \frac{2}{3}$ (where the latter probability is taken over the perturbation).}
\item{$|Q_i| \geq n^{\eps}$ with probability at least $\frac{1}{2}$.}
\end{enumerate}
\end{lemma}
\begin{proof}
The correctness of the algorithm and definition of $Q_i$ implies the first claim; i.e for a large fraction of the points in $S$, with constant probability the perturbed point generated $q_i = \nu_{\eps, \frac{\eps}{\mu}} (s_i)$ will be a \emph{good} point for reconstructing $p_i$. 

For the second claim observe that there are $\Omega(d)^{\Omega(\eps d)}$ points $q$ within at most $\eps d$ distance of $s_i$. Since our algorithm reconstructs $p_i$ with constant probability, the majority of these points must be good for $p_i$. Since 
$n \leq 2^d$, hence $n \leq d^d$ as well and we must have that
$|Q_i| \geq n^{\eps}$ with probability at least $\frac{1}{2}$ for sufficiently large $d$.
\end{proof}

Now define $A^t_j$ to be the set of $q \in 
\{0,1 \}^d$ such that $j = l_t(q)$. In the non adaptive domain we can assume that all cells 
are light; i.e. w.l.o.g for every $1 \leq j \leq m$ and $1 \leq t \leq r$ it holds that $|A^t_j| 
\leq \frac{2^d}{m}.$ The reason is that a cell $j$ for which $A^t_j$ is large (for some $t$) could 
be split into $|A^t_j| / \frac{2^d}{m}$ light cells with the total number of new cells bounded by $m$.  Our argument now analyzes \emph{shattered} points, which fall into any light cell with very low probability after perturbation.
\begin{defn}
A point $s_i$ is shattered if $\max_{j,t} [ \Pr \nu_{\eps, \frac{\eps}{\mu}} (s_i) \in A^t_j ] \leq n^{\frac{-c_1 \eps}{2}}$ for suitable choice of constant $c_1$.\end{defn}

Note first that $mw \geq \Omega(nd)$, just because the table needs enough space to store all the points to report. Since $w$ is upper bounded by a constant polynomial in $d$, we have that $m$ is polynomial in $n$.
Now for every non-adaptive algorithm, our isoperimetry bound implies that the probability over the choice of $s_i$ that $s_i$ is not shattered is at most $m^{\frac{-c_1 \eps}{2}}$ by Lemma~\ref{modularshatter}.  This probability is also at most $n^{-c_2 \eps}$ for suitable choice of constant $c_2$.   Hence with probability at least $\frac{1}{3}$, it holds that $s_i$ is shattered and  $|Q_i| \geq n^{\eps}$.  

We show that in such a case with constant probability there exists a $q^* \in Q_i$ , such that $l_{[r]} \subset L$, i.e. all 
the table lookups for point $q^*$ are contained in $T[L]$. Our procedure will simply sample points from $\nu_{\frac{\eps}{\mu}, \eps }(s_i)$  until it finds such a point $q^* \in Q_i$. Then by definition of good points, we can reconstruct $p$.

Since $s_i$ is shattered it holds that there are at most $r Q_i / n^{c_2 \eps}$ points in $Q$ that are mapped to the same cell. Now since $|Q_i| \geq n^{\eps}$ we have that there are at least $n^{c_2 \eps} / r$ different good $q$ which map into different rows of the table.  Each of these $q$ has all its probe locations map into $T[L]$ with probability at 
least $\frac{r}{n^{c_2 \eps}}$ so with probability $\geq \frac{1}{2}$ at least one point maps into $T[L]$ and we can reconstruct $p_i$ thereby.

\end{proof}

\section{Lower bounds via classical problems on the Hamming cube.}

In this section we will lay out lower bounds on the Bregman approximate near neighbor via reductions from the Hamming cube. The first reduction follows from ``symmetrizing'' our input by a simple bit trick and hence is independent of any asymmetry or $\mu$ parameter. The second reduction follows from the observation that for 
large enough $\mu$ ($\mu \geq cd $ for suitable constant $c$), only one direction of bit flip essentially determines the nearest neighbor. Under this regime, \ac{pm} can be reduced to our problem and hence the latter inherits the corresponding lower bounds. 
\subsection{A lower bound via $\ell_1$.}
\label{sec:reduct-from-hamm}

We start by defining a combinatorial structure isomorphic to the Hamming cube, which we term the \emph{pseudo-Hamming-cube}. 
\begin{defn}
  Given any uniform Bregman divergence $\breg : \reals^{2d} \times \reals^{2d} \to \reals$, a pseudo-Hamming-cube $C_{2 \phi} \subset \reals^{2d}$ is a set of $2^d$ points with a bijection $f_{ \phi} :\{0,1\}^d \to C_{\phi}$ and a fixed constant $c_0 \in \reals$ so that 
$\breg(f_{ \phi}(x),f_{ \phi}(y)) = c_0||x-y||_1$,
$\forall x,y \in \{0, 1\}^d$. 
\end{defn}
Given a specific uniform Bregman divergence $\breg$, we can compute a corresponding pseudo-Hamming cube explicitly.
\begin{lemma}
For any uniform $\breg : \reals^{2d} \times \reals^{2d} \to \reals$, there exists a pseudo-Hamming-cube $C_{\phi} \subset \reals^{2d}$ and a suitable constant $c_0 \in \reals$. 
\label{lem:pesudo}
\end{lemma}
\begin{proof}
 Recall that for uniform $\breg$, we have that $\phi(x) = \sum_{i=1}^d \phi_{\reals}(x_i)$. Pick arbitrary $a,b$ in the domain of $\phi_{\reals}$ and hence obtain the two ordered pairs $(a,b)$ and $(b,a)$ in $\reals^2$. Now stipulate $C_{ \phi} = \{ x_1 \times x_2 \ldots x_d \text{ s.t } x_i \in \{(a,b) ,
 (b,a)\} \} \subset \reals^{2d}$. Note that $|C_{\phi} = 2^d|$. We define the isomorphism $f_{\phi} :\{0,1\}^d \to C_{\phi}$ by using the helper function $\bar{f} : \{ 0,1 \} \to \{ (a,b) , (b,a) \}$
where $\bar{f}(0) = (a,b)$ and $\bar{f}(1) = (b,a)$. Now we state for $x \in \{0,1\}^d$ ,
 \begin{equation*}
f_{\phi}(x) =  \bar{f}(x_1) \times \bar{f}(x_2) \times \ldots \bar{f}(x_d).
\end{equation*}
The insight is that the component $D_{\phi_{\reals}}$ of $D_{\phi}$ on $\bar{f}(x_i)$ and $\bar{f}(y_i)$ between any two $x,y \in C_{\phi}$ is symmetrized, as
\begin{equation} D_{\phi_{\reals \times \reals}}((a,b), (b,a)) = D_{\phi_{\reals \times \reals}}((b,a), (a,b))
= D_{\phi_{\reals}}(b,a) + D_{\phi_{\reals}}(a,b). \end{equation}
Direct computation now shows that $\forall x,y \in \{0, 1\}^d$, we have that:
\begin{equation}
 \breg(f_{ \phi}(x),f_{ \phi}(y)) = \left(D_{\phi_{\reals}} \left(b,a \right) + D_{\phi_{\reals}}\left(a,b \right) \right) ||x-y||_1.
\end{equation} 
This completes our proof, with $c_0 = D_{\phi_{\reals}}(b,a) + D_{\phi_{\reals}}(a,b) $.
\end{proof}
Now that we have defined a distance preserving mapping from the $\ell_1$ Hamming cube in $\reals^d$ with a
$\breg$ pseudo-Hamming-cube in $\reals^{2d}$, it follows that LSH and ANN lower bounds for the $\ell_1$ Hamming cube now transfer over to $\breg$. In particular (and we list just a few here): 

\begin{itemize}
\item{A cell probe lower bound of $\Omega(\frac{d}{\log n})$ queries for any randomized algorithm that solves \emph{exact} 
nearest neighbor search on the Bregman cube in polynomial space and word size
polynomial in $d$, $\log n$ via \cite{emnn}.}
\item{A cell probe lower lower bound of $\Omega(d^{1-o(1)})$ queries for any
    deterministic algorithm that returns a constant factor approximation to the
    Bregman nearest neighbor via \cite{liu}.}
\item{A cell probe lower bound of $\Omega \left(\frac{\log \log d}{\log \log \log d} \right)$ for any randomized algorithm that returns 
a constant factor approximation to the Bregman nearest neighbor via \cite{chakraOptimal}.}
\end{itemize}

\subsection{A lower bound via \ac{pm}.}
\label{sec:partial}
We consider the following version of the \emph{\ac{pm} problem}: given point set $P \subset \{0,1 \}^d$ and $q \in \{0,1 \}^d$ determine whether 
$q$ dominates any point in $P$, i.e., output YES if $\exists p \in P$, s.t. $q_i \geq p_i$, $\forall 1 \leq i \leq d$ and NO otherwise. This is known by folklore to be equivalent to the more popular statement of the problem where $q \in \{0, 1 , * \}^d$ and we must determine whether $q$ matches any string in $P$ (and where $*$ can match anything). See for instance \cite{unifying} for a statement of this equivalence.

We construct our reduction from an instance of the \ac{pm} problem to an Bregman approximate near neighbor instance as follows. Set $\mu = 2d+1$, and define $D$ as given in Section \ref{sec:bregman-cube}. It is clear that $D(q,x) \leq d$ if and only if $x$ is a \ac{pm} for $q$ and that $D(q,x) \geq 2d+1$ otherwise. We immediately obtain the following:

\begin{theorem}\label{thm:from-partial}
Let $P \subset \{0,1 \}^d$, $q \in \{0,1 \}^d$ and $\mu \geq 2d+1$. Any algorithm which returns a $2$-approximate Bregman nearest neighbor $p \in P$ to query point $q$ with constant probability will solve the \ac{pm} problem with constant probability. 
\end{theorem}

This simple reduction immediately implies that for $\mu \geq \Omega(d)$ and a constant factor approximate nearest neighbor, our problem inherits the lower bounds on \ac{pm}.  These include, but are not limited to the following:

\begin{itemize}
\item{An $\Omega \left( \frac{d}{\log d} \right)$ lower bound on the number of cell probes for any randomized
 near-linear space algorithm, via a result by Patrascu and Thorup \cite{directsum}.}
\item{An $\Omega \left(2^{ \Omega \left( d / r \right)} \right)$ lower bound on the space required by any randomized algorithm which uses $r$ queries, via the result by Patrascu \cite{unifying}.}
\end{itemize}

\subsection{Comparisons and comments on the behavior of the lower bounds with $\mu$.}
It is worthwhile to contrast the bounds obtained from the simple reduction of Section \ref{sec:partial} to our more involved main result of Theorem \ref{thm:main}.

\begin{itemize}
\item{The simple reduction from \ac{pm} in Section \ref{sec:partial} implies a lower bound only for $\mu \geq \Omega(d)$, whereas our main result of Theorem \ref{thm:main} holds for far smaller asymptotic ranges of $\mu$, upto $O \left( \frac{d}{\log n} \right)$.}
\item{ 
For $\mu = \Omega \left( \frac{d}{ \log n} \right)$, Theorem \ref{thm:main} already implies a lower bound of $dn^{1+\Omega \left(d / r \log n \right) } =
2^{\Omega \left(\frac{d}{r} \right)}$ on the space required by a (non-adaptive) data structure that uses only $r$ queries. The reduction from \ac{pm} achieves this same space lower bound at a far higher $\mu = \Omega(d)$ via the result of \cite{unifying}.}
\item{For the strictly linear space ($O(nd)$) regime, Theorem 1.1 implies a lower bound of $\Omega(d)$ on number of queries required for our Bregman \ac{ann}. (Set 
$dn^{\Omega \left(1 + d / r \log n \right) } = O(nd)$, to get that $n^{\Omega \left(d / r \log n \right) } = 2^{\Omega(d/r)}$ must be $O(1)$ and hence $r = \Omega(d)$.) This is infact \emph{stronger} than any lower bound for number of queries known on \ac{pm} (the best we are aware of is $\Omega(d / \log d)$ by
~\cite{directsum} under any polynomial space).}
\end{itemize} 

As such, the parameter $\mu$ appears a natural way to interpolate between the well known problems of approximate nearest neighbor under $\ell_1$ and the \ac{pm} problem. At constant values of $\mu$, a constant factor ANN under $D$ corresponds to a constant factor ANN under $\ell_1$, whereas at $\mu \geq \Omega(d)$, a constant ANN under $D$ solves \ac{pm}. The point $\mu = \Omega(\frac{d}{\log n})$ then appears to be an interesting point along this interpolation where the space lower bounds are already asymptotically equal to those for \ac{pm}, with the qualifier that Theorem \ref{thm:main} holds for non-adaptive data structures.  The question remains open of whether \ac{pm} lower bounds themselves could infact be strengthened by a reduction to  Bregman \ac{ann}.

\section{Open Questions}
\label{sec:oq}
One open question is to convert our lower bounds to be non-adaptive. Panigrahy,
Talwar and Wieder~\cite{geometric} do give such a conversion for $\ell_2^2$
under symmetric perturbations, but it is unclear how to generalize their
argument to asymmetric perturbation operators. A second intriguing direction we
have already referred is whether our analysis of asymmetric isoperimetry can
open a different avenue of attack for lower bounds on Partial Match. Finally,
the question remains as to whether directed hypercontractivity offers insights
or generalizations for other expansion related problems previously considered on
the undirected hypercube. In this regard, recent work by Chakrabarty and
Seshadhri \cite{chakrabarty2013} and Khot \etal \cite{khot15:_monot_testin_boolean_isoper_theor} on formulations of directed hypercontractivity for
property testing represent a promising direction. 

\section{Acknowledgments}
\label{sec:acknowledgements}
The authors thank Piotr Indyk, Elchanan Mossel, James Lee, John Moeller, 
Amit Chakrabarti and anonymous referees for helpful discussions. We
also thank Nathan Keller for pointing us to the paper by Ahlberg \etal \cite{biasedcontract} and a reviewer for pointing us towards the reduction
from partial match.
\newpage

\bibliographystyle{plainnat}
\bibliography{pm,nn}

\newpage
\appendix
\section{Proof of Theorem \ref{cited}}
\label{sec:ahlberg}

As we mentioned in Section \ref{sec:hyperc-r_p-0}, \citet{biasedcontract}
establish a relation between biased and asymmetric noise operators Fourier
coefficients. The notation and setting they use are for a different application,
and so we give a self-contained reproduction of their proof here. 

Recall that our goal is to prove that
\begin{align*}
\widehat{R_{0,p}f}^{\left(1/2 \right)}(S) &= \left( \sqrt{\frac{1-p}{1+p}} \right)^{|S|} \hat{f}^{ \left(\frac{1-p}{2} \right)}(S).\\
\widehat{R_{p,0} f}^{\left( 1/2 \right)}(S) & = \left( \sqrt{\frac{1-p}{1+p}} \right)^{|S|} \hat{f}^{ \left(\frac{1+p}{2} \right)}(S).
\end{align*}

\begin{proof}
We prove the first equality. The second follows a similar argument, exchanging the measure of $0$ and $1$ alongwith the direction of perturbation.

First by definition:
\begin{equation}\label{starteq}
\widehat{R_{0,p}f}^{\left(1/2 \right)}(S) =E_{x \in \{0,1\}^d} \left[ E_{z|x} \left[  \chi_S(x) f(z) \right] \right]
\end{equation}

Now exchanging the order of integration over $x$ and $z$ simplifies analysis. Namely, we have that:

\begin{equation}\label{overall}
E_{x \in \{0,1\}^d} \left[ E_{z|x} \left[  \chi_S(x) f(z) \right] \right]  = E_{z \in \{0, 1 \}^d} \left[E_{x|z} \left[ \chi_S(x) f(z) \right] \right] 
=E_{z \in \{0, 1 \}^d} f(z) \left[E_{x|z} \left[ \chi_S(x) \right] \right]
\end{equation}

We define $E_{z \in \{0,1\}^d}$ more carefully, since this need not correspond to the uniform distribution. First note $\Pr[z_i = 0] = 
\Pr[x_i = 0]  + \Pr[x_i=1]* p = \frac{1+p}{2}$.  Similarly, note that $\Pr[z_i = 1] = \Pr[x_i=1]*(1-p) = \frac{1-p}{2}$. This immediately shows 
$E_{z \in \{0,1 \}^d }$ corresponds to the measure $\kappa_{\frac{1-p}{2}}$.

Now equation \ref{overall} lets us focus on analyzing the term $E_{x|z} \left[ \chi_S(x) \right]$ point wise. We observe that $\chi_S(x) = \prod_{i \in S} \chi_i(x)$ and that $\nu_{0,p}$ is applied independently on each bit. Hence we can reduce the complexity even further to analyzing essentially on $1$ bit:

\begin{equation}\label{theprod}
E_{x|z} \left[ \chi_S(x) \right] = \prod_{i \in S} E_{x|z} \chi_i(x).
\end{equation}
 
That is we focus on $E_{x|z} \chi_i(x)$ for any fixed $i$. Now by Bayes' rule:
\begin{equation}
\Pr[x_i = 0 | z_i =0] = \frac{\Pr[x_i=0, z_i=0]}{\Pr[z_i=0]}  =  \frac{1/2}{(1+p)/2} = \frac{1}{1+p}.
\end{equation}

And similarly:
\begin{equation} 
\Pr[x_i = 1 | z_i =0] = \frac{\Pr[x_i=1, z_i=0]}{\Pr[z_i=0]} =  \frac{p/2}{(1+p)/2} = \frac{p}{1+p}.
\end{equation}

So we then get
\begin{align*}E_{x|z_i=0} \chi_i(x) &= (1) \Pr[x_i=0 | z_i = 0]  + (-1) \Pr[x_i=1 | z_i = 0] \\
&= (1) \frac{1}{1+p} + (-1) \frac{p}{1+p}  = \frac{1-p}{1+p}.
\end{align*}
Or  \begin{equation}\label{zero} E_{x|z_i=0} \chi_i(x) = = \frac{1-p}{1+p} . \end{equation}

Now if $z_i = 1$ then we must have $x_i  = 1$, since the $\nu_{0,p}$ cannot flip $x_i = 0$ to $x_i = 1$. Recall that 
$\chi_i(x) = -1$ iff $x_i = 1$ and hence 
\begin{equation}\label{one}E_{x|z_i=1} \chi_i(x) = -1.
\end{equation}.

We now wish to normalize equations \ref{zero} and \ref{one}. That is, find a biased measure $p'$ and attenuating factor $\eps$ such that $E_{x|z} \chi_i(x) = \eps \chi_i^{p'}(z)$. This sets up the two equations, recalling that $\chi^{p'}(0) = \sqrt{\frac{p'}{1-p'}}$ and $\chi^{p'}(1) = - \sqrt{\frac{1-p'}{p'}}$.

\begin{equation}\label{first}
\frac{1-p}{1+p} =  \eps \sqrt { \frac{p'}{1-p'}}.
\end{equation}

\begin{equation}\label{second}
-1 =  -\eps \sqrt {\frac{1-p'}{p'}}
\end{equation}

Multiplying equations \ref{first} and \ref{second} gives $-\frac{1-p}{1+p} = -\eps^2$ or more simply $\eps = \sqrt{\frac{1-p}{1+p}}$. Substituting in equation \ref{second} for $\eps$ now gives $\sqrt{\frac{1-p'}{p'}} =  \sqrt{\frac{1+p}{1-p}}$. Clearly this is satisfied by putting $p' = \frac{1-p}{2}$. 
We finally obtain:
\begin{equation}\label{key}
E_{x|z} \chi_i(x) =  \sqrt{\frac{1-p}{1+p}} \chi_i^{\left( \frac{1-p}{2} \right)}(z).
\end{equation}

By substituting in equation \ref{theprod}, we have:
\begin{equation}
E_{x|z} \left[ \chi_S(x) \right] =  \sqrt{\frac{1-p}{1+p}}^{|S|} \chi_S^{\left( \frac{1-p}{2} \right)}(z).
\end{equation}

 Substituting back in equation \ref{overall} now gives us:

\begin{equation} E_{z \in \{0, 1 \}^d} f(z) \left[E_{x|z} \left[ \chi_S(x) \right] \right] = E_{z \in \{0, 1 \}^d} f(z)  \sqrt{\frac{1-p}{1+p}}^{|S|} \chi_S^{\left( \frac{1-p}{2} \right)}(z) .
\end{equation}
Recalling that we have already shown $E_{z \in \{0 , 1 \}^d}$ corresponds to $\kappa_{\frac{1-p}{2}}$, the desired result now follows from definition of $\hat{f}^{\frac{1-p}{2}}$.
\end{proof}

\end{document}